\documentclass[a4paper,UKenglish]{lipics-v2018}

\usepackage{microtype}
\usepackage{amsmath}
\usepackage{amssymb}
\usepackage{mathrsfs}
\usepackage{graphicx}
\usepackage{color}
\usepackage[all]{xy}

\nolinenumbers

\newtheorem{proposition}{Proposition}

\renewcommand{\phi}{\varphi}

\newcommand{\blue}[1]{{#1}}
\newcommand{\N}{\mathbb{N}}
\newcommand{\Z}{\mathbb{Z}}
\newcommand{\Q}{\mathbb{Q}}

\newcommand{\ignore}[1]{}

\DeclareMathOperator{\End}{End}
\DeclareMathOperator{\Csp}{CSP}
\newcommand\tuple[1]{\mathbf{#1}}
\newcommand\mA{\mathbb A}
\newcommand\mB{\mathbb B}

\title{The complexity of disjunctive linear Diophantine constraints}

\author{Manuel Bodirsky}{Institut f\"ur Algebra, TU Dresden, Germany}
{manuel.bodirsky@tu-dresden.de}
{}
{Manuel Bodirsky has received funding from the ERC under the European Community's Seventh Framework Programme (Grant Agreement no. 681988, CSP-Infinity), and the DFG-funded project `Homogene Strukturen, Bedingungserf\"ullungsprobleme, und topologische Klone' (Project number 622397)} 
\author{Barnaby Martin}{Department of Computer Science, Durham University, U.K.}{barnabymartin@gmail.com}{}{}
\author{Marcello Mamino}{Dipartimento di Matematica, largo Pontecorvo 5, 56127 Pisa, Italy}
{marcello.mamino@dm.unipi.it}
{}
{Marcello Mamino has received funding from the ERC under the European Community's Seventh Framework Programme (Grant Agreement no. 681988, CSP-Infinity).}
\author{Antoine Mottet}{Institut f\"ur Algebra, TU Dresden, Germany}{antoine.mottet@tu-dresden.de}{}{Supported by the DFG Gratuiertenkolleg 1763 (QuantLA).}

\authorrunning{M. Bodirsky et al.}

\Copyright{M. Bodirsky et al.}

\subjclass{F.2.2 Nonnumerical Algorithms and Problems}
\keywords{Constraint Satisfaction, Presburger Arithmetic, Computational Complexity}

\begin{document}

\maketitle

\begin{abstract}
We study the Constraint Satisfaction Problem CSP(\mbox{$\mA$}), where $\mA$ is first-order definable in $(\Z;+,1)$ and contains $+$. We prove such problems are either in P or NP-complete.
\end{abstract}

\section{Introduction}

A \emph{constraint satisfaction problem} (CSP) is a computational problem where the input consists of a finite set of variables and a finite set of constraints, and where the question is whether there exists a mapping from the variables to 
some fixed domain such that all the constraints are satisfied. 
When the domain is finite, and arbitrary constraints are permitted in the input, the CSP is NP-complete. 
However, when only constraints for a restricted set of relations are allowed in the input, it might be possible to solve the CSP in polynomial time. 
The set of relations that is allowed to formulate the constraints in the input is often called the \emph{constraint language}. The question as to which constraint
languages give rise to polynomial-time solvable CSPs
has been the topic of intensive research over the past years. It was conjectured by Feder and Vardi~\cite{FederVardi} that CSPs for constraint languages over finite domains have a complexity dichotomy: they are in P or are NP-complete. This conjecture has recently been proved \cite{BulatovFVConjecture,ZhukFVConjecture}.

A famous CSP over an infinite domain is the feasibility question for Integer Programs.
It is of great importance in practice and theory of computing, and NP-complete. 
In order to obtain a systematic understanding
of polynomial-time solvable restrictions and variations of this problem, 
Jonsson and L\"o\"ow~\cite{JonssonLoow}  proposed to study the class of
CSPs where the constraint language $\mA$ 
is definable in 
\emph{Presburger arithmetic}; that is,
consists
of relations that have a first-order definition over $({\mathbb Z};<,+,1)$.
Equivalently, each relation $R(x_1,\dots,x_n)$ in $\mA$ can be
defined by a disjunction of conjunctions of the atomic formulas of the form $p \leq 0$ where
 $p$ is a linear polynomial with integer coefficients and  variables from $\{x_1,\dots,x_n\}$. 
The constraint satisfaction problem for $\mA$, denoted by $\Csp(\mA)$, is the problem of deciding whether a
given conjunction of formulas of the form $R(y_1,\dots,y_n)$, for some $n$-ary $R$ from $\mA$, is satisfiable in $\mA$. 
By appropriately choosing such a constraint language $\mA$, a great
variety of problems over the integers can be formulated as $\Csp(\mA)$.
Several constraint languages $\mA$ over the integers are known where the CSP can be solved in polynomial time. Among the most famous of these is Linear Diophantine Equations, namely $\Csp(\Z;+,1)$.
The first polynomial-time algorithms for the satisfiability of linear Diophantine equation systems have been discovered by Frumkin and, independently, Sieveking and von zur Gathen.  Kannan and Bachem~\cite{KannanBachem} presented a method based on first computing the Hermite Normal Form of the matrix given by the linear system (see discussion in the text-book of Schrijver~\cite{Schrijver}). Further improvements have been made in~\cite{GaussElimPoly,Storjohann,MW2001}.
In the present parlance, $\Csp(\Z;<,+,1)$ is Integer Program feasibility itself. 
However, a complete complexity classification for the CSPs of Jonsson-L\"o\"ow languages appears to be a very ambitious goal. 

Among the classes of constraint language that fall into the framework of Jonsson and L\"o\"ow
are the \emph{distance CSPs} of \cite{BodDalMarMotPin,BodMarMot} and the \emph{temporal CSPs} of \cite{tcsps-journal}. Temporal CSPs are those whose constraint language is  first-order definable in $(\Q;<)$ and \emph{discrete temporal CSPs} are those whose constraint languages is first-order definable in $(\Z;<)$. The classification for discrete temporal CSPs represents the join of the work on temporal CSPs and distance CSPs, and has only recently been accomplished \cite{dCSPs3}.

Moving away from the discrete and non-dense, $(\Q;<)$ is not the only structure for which constraint languages that are first-order expansions have had their CSPs classified. The situation for such expansions of the language of linear programming, $(\Q;<,+,1)$ was settled in \cite{Essentially-convex}. Perhaps, more interesting for us is the simplified situation in which only first-order expansions of $(\Q;+)$ are considered, in \cite{HornOrFull}. Most recently, the work \cite{JonssonThapperJCSS16} delivers a classification for all first-order definitions in $(\Q;<,+,1)$ that contain $+$, thus properly extending the result from \cite{Essentially-convex}. In these works, the class of relations quantifier-free definable in Horn CNF plays a key role. In this context, the atomic relations are inequalities and equalities, and each clause may have no more than one equality or inequality. That is, additional disjuncts in clauses must be disequalities. For first-order expansions of $(\Q;+)$, the tractable constraint languages are precisely those that are quantifier-free Horn definable on $(\Q;+)$ \cite{HornOrFull}.

However, the integers behave very differently from the rationals or reals and even simple types of Horn definitions engender intractable constraint languages, as documented in \cite{JonssonLoow}. This article shows, depending on one's perspective, [un]surprisingly, that the tractability frontier for first-order definitions of $(\Z;+,1)$, containing $+$, coincides with that for first-order expansions of $(\Q;+)$.
Under a mild technical assumption on $\mA$,
either all of its relations are quantifier-free Horn definable, in the expansion of $(\Z;+,1)$ associated with its quantifier elimination, and CSP$(\mA)$ is solvable in P; or CSP$(\mA)$ is NP-complete.
From this we obtain the following dichotomy result. 

\begin{theorem}\label{thm:dichotomy}
Let $\mA$ be an expansion of $(\Z;+)$ by \blue{finitely many} relations
with a first-order definition in $(\Z;+,1)$. Then
$\Csp(\mA)$ is in P or NP-complete. 
\end{theorem}

\subsection*{Other related work}
This work forms part of a growing body addressing infinite-domain CSPs. One line of that work concerns $\omega$-categorical and finitely-bounded constraint languages and the other line considers constraint languages over ordinary structures of arithmetic. The two lines overlap in the foundational work on temporal CSPs \cite{tcsps-journal}. The outstanding other result in the first line is \cite{BodPin-Schaefer} and recent progress can be seen in \cite{BartoP16,BartoKOPP17}. The importance of the latter line is discussed in the survey \cite{NumericCSPs}.

The CSP for certain finite groups were studied already in the seminal \cite{FederVardi}. $(\Z;+,0)$ is a group \emph{par excellence} and our work takes inspiration from that paper. One of our hardness results uses its Theorem 34 and our tractable cases include the situation when all relations are subgroups, or cosets of subgroups, of powers of $\Z$ (\mbox{cf.} \cite{FederVardi}, Theorem 33). However, not all first-order expansions of $(\Z;+,1)$ are related to groups, and we have other sources of tractability too. 

\section{Preliminaries}
We say a relational structure $\mA$ is \emph{first-order definable in $(\Z;+,1)$} (or a \emph{first-order reduct of$(\Z;+,1)$}  if it is over domain $\Z$ with relations specified by first-order formulas over $(\Z;+,1)$. 
An \emph{endomorphism} of $\mA$ is a map $h\colon\Z\to\Z$ such that for every relation $R$ of $\mA$ and every tuple $(a_1,\dots,a_k)\in \Z^k$, we have $\tuple a\in R\Rightarrow h(\tuple a)\in R$.
We say that $h$ is a \emph{self-embedding} if the implication is an equivalence.

A formula over a relational signature $\sigma$ is \emph{primitive positive (pp)}
if it is of the form $\exists x_1,\dots,x_k (\psi_1 \wedge \cdots \wedge \psi_m)$ where each
$\psi_i$ is an atomic relation built from $\sigma$. Note that $0$ is pp-definable in $(\Z;+)$. A \emph{sentence} is a formula without free variables. 

The \emph{constraint satisfaction problem} for a \blue{structure $\mA$ with finite relational signature $\sigma$}, denoted $\Csp(\mA)$, is the following computational problem. 

\medskip

{\bf Input:} A primitive positive $\sigma$-sentence $\Phi$.\\
{\bf Question:} $\mA \models \Phi$?

\medskip
All CSPs will be defined over strictly relational signatures, thus in this context $+$ must be considered a ternary relation and $1$ a constant or singleton unary relation, depending on taste. Since we also use $+$ with its common meaning of binary operation, we concede guilt for overloading. However,  the two uses will never conflict in meaning, so we will not dwell further on the matter.
\blue{If $\mA$ is first-order definable in $(\Z;<,+,1)$ then $\Csp(\mA)$ is in NP (this is noted e.g.\ in \cite{JonssonLoow}).}

    A \emph{linear equation} is a formula of the form $\sum_{i=1}^n a_ix_i=b$ with $a_1,\dots,a_n,b\in\Z$, whose free variables are $\{x_1,\dots,x_n\}$.
    A \emph{modular linear equation} is a formula of the form $\sum_{i=1}^n a_ix_i=b\bmod c$ with $a_1,\dots,a_n,b,c\in\Z$.
Let $\mathcal L_{(\Z;+,1)}$ be the infinite relational language containing a relation symbol for each  linear equation and modular linear equation.
For convenience, we consider first-order logic to have native symbols for $\top$ (true) and $\bot$ (false).
It is well-known that $(\Z;+,1)$ admits quantifier elimination in the language $\mathcal L_{(\Z;+,1)}$ (see \cite{Presburger}, or \cite[Corollary 3.1.21]{Marker} for a more modern treatment).
Call an $\mathcal L_{(\Z;+,1)}$-formula \emph{standard} if it does not contain a negated modular linear equation.
Every $\mathcal L_{(\Z;+,1)}$-formula is equivalent  to a standard $\mathcal L_{(\Z;+,1)}$-formula, since a negated modular linear equation
is equivalent to a disjunction of modular linear equations (i.e., $k\neq b\bmod c\Leftrightarrow \bigvee_{0\leq a\leq c, a\neq b} k=a\bmod c$).
We say that an equation \emph{appears} in a formula if it is a positive or negative literal in that formula.

Any subgroup $G$ of $\Z^k$ can be given by a finite set of \emph{generators}, i.e., $k$-tuples $\tuple g^1,\dots,\tuple g^m$,
such that for every $\tuple g\in G$, there are $\lambda_1,\dots,\lambda_m\in\Z$ such that $\tuple g= \sum_i\lambda_i\tuple g^i$,
where we write $\lambda\cdot\tuple g$ for $(\lambda g_1,\dots,\lambda g_k)$.
A \emph{coset} of a subgroup $G$ of $\Z^k$ is any set of the form $\tuple a+G:=\{\tuple a+\tuple g \mid \tuple g\in G\}$, where $\tuple a\in\Z^k$.
By moving to a standard formula, we are in a position to deduce the following.

\begin{proposition}\label{prop:easy}
Suppose $R$ is a unary relation first-order definable in $(\Z;+,1)$.
Then $R$ has the form $(R^{\circ} \cup R^+)\setminus R^-$,
where $R^{\circ}$ is a finite union of cosets of nontrivial subgroups of $\Z$, and $R^+$ and $R^-$ are finite disjoint sets of integers.
\label{prop:QE}
\end{proposition}
\begin{proof}
	Consider a disjunction $\phi$ of equations (possibly negated and modular equations). If this disjunction contains a negated equation $ax\neq c$, then $\phi$ defines a relation that contains $\Z\setminus\{c/a\}$
	and is therefore as in the statement.
	Otherwise, $\phi$ contains only positive linear equation and modular equations, and the relation that $\phi$ defines is clearly of the form $R^{\circ} \cup R^+$ for some finite set $R^+$ and some union $R^{\circ}$ of nontrivial subgroups of $\Z$.
	
	Consider a quantifier-free formula $\phi$ in conjunctive normal form defining $R$. Each conjunct defines a relation of the right form, per the previous paragraph. It is easily checked that a conjunction
	of relations of this form is again a relation of the form $(R^{\circ} \cup R^+)\setminus R^-$, so that we have proved that every quantifier-free formula with one free variables defines a relation of the right form.
	The proposition then follows from quantifier-elimination.
\end{proof}

\blue{Note that if 
$R^+ \cap R^\circ=\emptyset$ and $R^- \subset R^\circ$, then $R^+$, $R^-$, and $R^{\circ}$ are unique. We use the terminology with this convention for all
unary relations $R$ that are first-order definable in $(\Z;+,1)$ throughout the article.}


\begin{definition}
    Let $\phi$ be an $\mathcal L_{(\Z;+,1)}$-formula.
    We say that $\phi$ is \emph{Horn} if it is a conjunction of clauses of the form
    $$\bigvee_{i=1}^n \neg\phi_i \lor \phi_0$$
    where $\phi_1,\dots,\phi_n$ are linear equations and $\phi_0$ is a linear or a modular linear equation.
\end{definition}

\begin{example}
Singletons, cofinite unary relations, and cosets of subgroups of $\Z^n$ are examples of Horn-definable relations.
\end{example}
\ignore{\begin{theorem}\label{thm:main}
    Let $\mA$ be quantifier-free definable in $\mathcal{L}_{(\Z;+,1)}$ and suppose that contains $+$ and is a core. 
    Let $\mB$ be a core of $\mA$ and suppose that $\mB$ is infinite.
    Then $\Csp(\mA)$ is NP-complete if the relations of $\mA$ are not definable by Horn formulas,
    and is in P otherwise.
\end{theorem}}

\section{Cores}

If $\mA$ is a first-order expansion of $(\Z;+)$, note that its endomorphisms are precisely of the form $x \mapsto \lambda x$ for some $\lambda\in \Z$.
Therefore, we view in the following $\End(\mA)$ as a subset of $\Z$, where the monoid structure on $\End(\mA)$ implies that as a subset of $\Z$,
it is closed under multiplication and contains $1$.
We say that $\mA$ is a core if all its endomorphisms are self-embeddings, and that $\mB$ is a core of $\mA$ if $\mA$ and $\mB$ are homomorphically equivalent and $\mB$ is a core.

\begin{lemma}\label{lem:core}
    Let $\mA$ be first-order definable in $(\Z;+,1)$, and suppose that $\mA$ contains $+$.
    There exists a structure which is a core of $\mA$,  and which is either
    a 1-element structure or first-order definable in $(\Z;+,1)$ and containing $+$.
\end{lemma}
\begin{proof}
    If $0\in \End(\mA)$ then the lemma is clearly true ($\mA$ being homomorphically equivalent to the substructure of $\mA$ induced by $\{0\}$), so let us assume that $0\not\in \End(\mA)$.
    Similarly we can assume that $\End(\mA)\not\subseteq\{-1,1\}$, otherwise $\mA$ is already a core.
    For a quantifier-free formula $\psi$ and an integer $\lambda$, define $\psi/\lambda$ by induction on $\psi$ as follows:
    \begin{itemize}
        \item if $\psi$ is $\sum \lambda_ix_i=c$ and $\lambda$ divides $c$, then $\psi/\lambda$ is $\sum \lambda_ix_i=c/\lambda$,
        \item if $\psi$ is $\sum \lambda_ix_i=c$ and $\lambda$ does not divide $c$, then $\psi/\lambda$ is $\bot$,
        \item if $\psi$ is $\sum \lambda_ix_i = c\bmod d$ and $\ell:=\gcd(\lambda,d)$ divides $c$, then $\psi/\lambda$ is $\sum \lambda_ix_i= ec/\ell\bmod d/\ell$ where $e$ is the inverse of $\lambda/\ell$ modulo $d/\ell$,
         \item if $\psi$ is $\sum \lambda_ix_i = c\bmod d$ and $\ell:=\gcd(\lambda,d)$ does not divide $c$, then $\psi/\lambda$ is $\bot$,
        \item extend to boolean combinations in the obvious fashion.
    \end{itemize}
    Note that for every tuple $\tuple a$, we have that $\tuple a$ satisfies $\psi/\lambda$ iff $\lambda\cdot\tuple a$ satisfies $\psi$.
    Indeed, if $\psi$ is a linear equation then this is clear. Similarly, it is clear if $\psi$ is a modular equation and $\ell:=\gcd(\lambda,d)$ does not divide $c$.
    Suppose that $\psi$ is a modular equation and $\ell:=\gcd(\lambda,d)$ divides $c$. If $\sum \lambda\lambda_ix_i=c\bmod d$ then $\lambda/\ell\cdot(\sum \lambda_ix_i) = qd/\ell + c/\ell$
    so that $e\lambda/\ell\cdot(\sum \lambda_ix_i) = (eq)\cdot d/\ell + ec/\ell$, where $e$ is the inverse of $\lambda/\ell$ modulo $d/\ell$ and $q\in \Z$.
    We therefore obtain $\sum \lambda_ix_i = ec/\ell\bmod d/\ell$.
    Conversely if $\sum \lambda_ix_i = ec/\ell\bmod d/\ell$ then $\sum \lambda\lambda_ix_i = (\lambda e)c/\ell + (\frac{\lambda}{\ell}q)d = c\bmod d$.
    
    Let $\psi$ be any quantifier-free $\mathcal L_{(\Z;+,1)}$-formula and suppose that $|\lambda|>1$.
    The only cases where some magnitudes of the integers on the right-hand sides of terms in the formula $\psi$ do not decrease by forming $\psi/\lambda$ is when $\psi$
    only contains literals either of the form $\sum \lambda_ix_i=0$ or of the form $\sum \lambda_ix_i = c\bmod d$ with $\lambda$ and $d$ coprime.
    Therefore, the sequence $\psi_0,\psi_1,\psi_2,\dots$ where $\psi_0$ is $\psi$ 
    and where $\psi_{i+1}$ is $\psi_i/\lambda$ for some $\lambda\in \End(\mA)$ with $|\lambda|>1$ reaches in a finite number of steps a fixpoint where all the literals
    are either of the form $\sum \lambda_ix_i=0$ or are modular equations whose modulus $d$ is such that $\lambda$ and $d$ are coprime.
    Let $n\geq1$ be such that for every $\psi$ defining a relation of $\mA$, the formula $\psi_n$ is a fixpoint.
    Let $\mB$ be the structure whose domain is $\Z$ and whose relations are $+$ and the relations defined by $\psi_n$
    for each $\psi$ defining a relation of $\mA$.
   
    We claim that $\mB$ is homomorphically equivalent to $\mA$ and is a core.
    The first claim is clear, since $\mB$ is isomorphic to the structure obtained from $\mA$ by successive applications of endomorphisms $x\mapsto \lambda\cdot x$ (in particular $\mB$ embeds into $\mA$).
    Let now $x\mapsto \lambda\cdot x$ be an endomorphism of $\mB$, and suppose that $\tuple a$ is a tuple in a relation $R$ of $\mB$.
Then we have that $\lambda \cdot \tuple a$ in $R$ since $x\mapsto \lambda\cdot x$ is an endomorphism. Conversely, note that $\lambda$ is coprime to $d$ or else we would not have reached a fixed point in the previous stage.
Thus, $\lambda^{\phi(d)}=1 \bmod d$, where $\phi(d)$ here is the totient of $d$.
It follows then that $\lambda^{\phi(d)}\tuple a =\tuple a \bmod d$. Suppose $\lambda\tuple a \in R$, then by applying $\phi(d)-1$ times an endomorphism, we derive $\lambda^{\phi(d)}\tuple a \in R$. It follows that $\tuple a \in R$, for both the cases that atoms are of the form $\sum \lambda_ix_i=0$ or are modular equations whose modulus $d$ is such that $\lambda$ and $d$ are coprime. Hence,  $x\mapsto \lambda\cdot x$ is an embedding of $\mA$.
\end{proof}

\blue{We order the standard formulas lexicographically with respect to (in this order)} 
\begin{enumerate}
	\item \blue{the} number of non-Horn clauses,
	\item \blue{the number of literals in clauses with at least two literals,}
	\item \blue{the number of all literals}, and 
	\item \blue{the sum of the absolute values} of all numbers appearing in an equation. 
\end{enumerate}
\blue{This order is used in a number of statements and proofs throughout the text, e.g., in Proposition~\ref{prop:syntactic-core}, 
Lemma~\ref{lem:step-1}, 
and Theorem~\ref{thm:not-Horn-hard}. 
 }
A \blue{standard} formula is \emph{minimal} if no smaller formula is equivalent to it. 

The following properties follow from the construction of cores in the previous proof.
\begin{proposition}\label{prop:syntactic-core}
	Let $\mA$ be first-order definable in $(\Z;+,1)$, and suppose that $\mA$ contains $+$ and is a core. Let $\lambda\in\End(\mA)$.
	Let $R$ be a relation of $\mA$ and let $\phi$ be a minimal standard formula defining $R$.
    \begin{itemize}
        \item If $\sum\lambda_ix_i=c$ is a linear equation appearing in $\phi$, then $c=0$ or $|\lambda|=1$.
    \item If $\sum\lambda_ix_i=c\bmod d$ is a modular linear equation in $\phi$, then $\lambda$ and $d$ are coprime.
    \end{itemize}
   Moreover, if $\End(\mA)=1+d\Z$ for some $d\geq 2$, then every relation of $\mA$
    can be expressed with a minimal formula in which all modular linear equations are modulo a divisor of $d$.
\end{proposition}
\begin{proof}
    The two items are clear from the proof of Lemma~\ref{lem:core}.
	For the last statement, let $d'$ be a modulus appearing in a minimal definition of a relation of $\mA$.
	By the second item, we have that $d'$ and $1+kd$ are coprime, for all $k\in\Z$.
	Let $\ell$ be such that $\ell d= -1 \bmod\frac{d'}{\gcd(d,d')}$.
	If $d'$ and $1+\ell d$ are coprime, there exist $u,v\in\Z$ such that $ud' + v(1+\ell d) = 1$.
	Taking this equation modulo $\frac{d'}{\gcd(d,d')}$ we obtain $0=1\bmod  \frac{d'}{\gcd(d,d')}$,
	so that $\gcd(d,d')=d'$ and $d'$ divides $d$.
\end{proof}

\section{Hardness}

\def\NAE{\textsf{Not-All-Equal-3-SAT}}
\def\OneInThree{\textsf{1-in-3-SAT}}

Our sources of hardness come from \emph{pp-interpretations}, that we define now.
A structure $\mB$ is said to be \emph{one-dimensional pp-interpretable} in $\mA$ if there exists a partial surjective map $h\colon A\to B$, called the \emph{coordinate map},
such that the inverse image of every relation of $\mB$ (including the equality relation and the unary relation $B$) under $h$ has a pp-definition in $A$.
Formally, we require that for every $k$-ary relation $R$ of $\mB$, there exists a pp-formula $\phi_R(x_1,\dots,x_k)$ in the language of $\mA$
such that
\[ \mA\models\phi_R(a_1,\dots,a_k) \Leftrightarrow \mB\models R(h(a_1),\dots,h(a_k))\]
holds for all $a_1,\dots,a_k\in A$. \blue{This requirement for the equality relation of $\mB$ and the unary relation $B$ implies} that the kernel of $h$ and its domain have a pp-definition in $\mA$.
It is well-known that if $\mB$ is pp-interpretable in $\mA$,
then $\Csp(\mB)$ reduces in polynomial time to $\Csp(\mA)$.

\subsection{The fully modular case}
\label{sect:fully-mod}
One of the sources of hardness for our problems are expansions of the \emph{general subgroup problem} from~\cite{FederVardi}.
The general subgroup problem of a finite abelian group $G$ is the CSP of $(G;+)$ expanded with a $k$-ary relation for every coset $\tuple a+ H$,
where $H$ is a subgroup of $G^k$.
It is known that this problem is solvable in polynomial time (under some reasonable encoding of the input); in modern parlance, this follows from the fact that the operation $(x,y,z)\mapsto x-y+z$
is a Maltsev polymorphism of the template.
Feder and Vardi~\cite[Theorem 34]{FederVardi} proved that the problem becomes NP-hard if the template is further expanded by any other relation.
 
The general subgroup problem of $\Z/d\Z$ can be viewed as a CSP of a first-order reduct of $(\Z;+,1)$
whose relations are defined by quantifier-free formulas only containing modular linear equations.
This motivates the following definition.

\begin{definition}
A relation $R\subseteq\Z^k$ is called \emph{fully modular} if it is definable by a conjunction of disjunctions of modular linear equations,
in which case we can even assume that all the modular linear equations involved in such a definition of $R$ have the same modulus $d\geq 1$.
\end{definition}

\begin{proposition}\label{prop:hardness-fully-modular}
    \blue{Let $\mA$ be a finite-signature core which is first-order definable in $(\Z;+,1)$ and contains $+$.} 
    Suppose that $\mA$ has a fully modular relation that is not Horn-definable.
    Then $\Csp(\mA)$ is NP-complete.
\end{proposition}
\begin{proof}
Let $R$ be a relation of $\mA$ that is not Horn-definable and fully modular, and let $d\geq 1$ be such that $R$ can be defined with only linear equalities modulo $d$.
\blue{Let $\mA/d\mA$ be the structure with domain $\Z/d\Z$ containing the ternary relation $+$
as well as a relation $S'$ for every relation $S$ of arity $k$ of $\mA$,} defined by
\[S'=\{(a_1,\dots,a_k) \mid \exists q\in\Z : (qd+a_1,\dots,qd+a_k)\in S\}.\]
Note that $\mA/d\mA$ is pp-interpretable in $\mA$: the coordinate map is the canonical projection $x\mapsto x\bmod d$,
whose kernel is pp-definable by the formula $\phi_=(x,y) := \exists z(x-y=dz)$.
As a consequence, $\Csp(\mA/d\mA)$ reduces in logarithmic space to $\Csp(\mA)$.
Moreover, if $\mA$ is a core then $\mA/d\mA$ is also a core.
It follows from general principles~\cite[Proposition 3.3]{wonderland} that $\Csp(\mA/d\mA,1)$ reduces to $\Csp(\mA/d\mA)$ and so to $\Csp(\mA)$.
Note that every coset of a subgroup of \blue{$(\Z/d\Z)^k$} is pp-definable in $(\mA/d\mA,1)$ and that if $R$ is not Horn-definable
then $R'$ is not a coset of a subgroup. It follows from Theorem 34 in the bible~\cite{FederVardi} that $\Csp(\mA)$ is NP-complete.
\end{proof}

\subsection{The unary case}
\def\Family{\{S_\lambda\}_{\lambda\in\Lambda}}
\label{sect:unary}
In order to prove Theorem~\ref{thm:dichotomy}, we now focus on the case of \emph{parametrised unary relations}.

\begin{definition}[Compatibility]
    Let \blue{$\Lambda\subseteq\Z \setminus \{0\}$} be a set containing $1$.
    We say that a set $\{S_\lambda\}_{\lambda\in\Lambda}$ of subsets of $\Z$
    \blue{that are definable in $(\Z;+,1)$} 
     is \emph{compatible}
    if there exist disjoint finite sets $A,B\subseteq\Z$ such that 
    \begin{itemize}
    \item $S_\lambda = (S^\circ_\lambda\cup \lambda\cdot A)\setminus \lambda \cdot B$ for all $\lambda\in\Lambda$
    and 
    \item for all $d\geq 1$ and $c\in\{0,\dots,d-1\}$, we have $c+d\Z\subseteq S^\circ_1\Leftrightarrow \lambda c+d\Z\subseteq S^\circ_\lambda$.
    \end{itemize}
    \end{definition}
    
 \begin{definition}[Uniform pp-definability]
   Let $\mA$ be a first-order reduct of $(\Z;+,1)$.  
    We say that $\{S_\lambda\}_{\lambda\in\Lambda}$ is \emph{uniformly pp-definable in $\mA$} if there exists a pp-formula $\theta(x,y)$
    such that $a\in S_\lambda$ if, and only if, $\mA\models \theta(\lambda,a)$.
\end{definition}

Note that the definition of being uniformly pp-definable implies that $\Lambda$ has a pp-definition in $\mA$, for $\exists y. \,\theta(x,y)$ is a pp-definition.
\blue{Let $S\subseteq\Z^2$ be a binary relation that is pp-definable in $\mA$. 
Then the family $\Family$
where $\Lambda := \{a \in \Z \mid (a,b) \in S  \text{ for some } b \in \Z \} \subseteq \Z \setminus \{0\}$ and $S_\lambda := \{a\in\Z\mid (\lambda,a)\in S\}$} 
is uniformly pp-definable in $\mA$.
\blue{But even if $S$ contains a tuple of the form $(1,b)$ and 
no tuple of the form $(0,b)$, it might not necessarily satisfy the compatibility 
condition, as illustrated in the following example.}
\begin{example}
\blue{Let $S=\{(a,b)\in\Z^2\mid a \neq 0 \wedge (a=b\lor a=2b)\}$. Then $\Lambda = \Z \setminus \{0\}$, and for $\lambda \in \Lambda$ we have 
$S_\lambda = \{\lambda\}$ if $\lambda=1\bmod 2$ and $S_\lambda = \{\lambda,\frac{\lambda}{2}\}$
if $\lambda=0\bmod 2$.} Therefore, the compatibility condition is not satisfied by 
$\Family$.
\end{example}

In the following proof, we write \OneInThree\ for $\Csp(\{0,1\};\{(1,0,0),(0,1,0),(0,0,1)\})$.
It is well-known that this problem is NP-complete  \blue{(\cite{Schaefer}; for a proof see~\cite{Papa})}.

\begin{lemma}\label{lem:finite-nontrivial-hard}
    Let $\mA$ be a \blue{finite-signature} first-order reduct of $(\Z;+,1)$ containing $+$.
    If $\{S_\lambda\}_{\lambda\in\Lambda}$ is a compatible set of unary relations that is uniformly pp-definable in $\mA$
    and if $1<|S_\lambda|<\infty$ for all $\lambda\in\Lambda$, then $\Csp(\mA)$ is NP-hard.
\end{lemma}
\begin{proof}
    Since every $S_\lambda$ is finite, one sees that $S_\lambda = \lambda\cdot A$ for the finite set $A$
    coming from the compatibility condition.
    Let $m_1:=\min(A)$ and $m_2:=\min(A\setminus\{m_1\})$.
    The formula \[\exists\lambda(x+y+z=(m_2-m_1)\lambda\land x+m_1\lambda\in S_\lambda \land y+m_1\lambda\in S_\lambda\land z+m_1\lambda \in S_\lambda \land \lambda\in\Lambda)\]
    defines the ternary relation consisting of $(a,b,c)\in\Z^3$ such that $a,b,c\in\{0, m_2-m_1\}$ and exactly one of $a,b,c$ is equal to $m_2-m_1$.
    Note that this formula is in the language of $\mA$, since $\{S_\lambda\}_{\lambda\in\Lambda}$ is uniformly pp-definable and in particular $\Lambda$ is pp-definable in $\mA$.
    This gives an interpretation of \OneInThree\ in $\mA$, using the map $h\colon \{0,m_2-m_1\}\to \{0,1\}$ such that $h(0)=0$ and $h(m_2-m_1)=1$.
    Therefore, $\Csp(\mA)$ is NP-hard.
\end{proof}

\begin{proposition}\label{prop:unary-hard}
    \blue{Let $\mA$ be a finite-signature first-order reduct of $(\Z;+,1)$ that contains $+$ and is a core.
    Let $\Family$ be a compatible family that is uniformly pp-definable in $\mA$ such that
    for every $\lambda\in\Lambda$ the set $S_\lambda$ is not Horn-definable. Then $\Csp(\mA)$ is NP-hard.}
\end{proposition}
\begin{proof}
    Let $A,B\subset\Z$ be finite such that $S_\lambda = (S^\circ_\lambda\cup \lambda\cdot A)\setminus(\lambda\cdot B)$ \blue{for all $\lambda \in \Lambda$}.
    \blue{Since $S_\lambda$ is not Horn-definable,
    we have $|S_\lambda|>1$ for all $\lambda  \in \Lambda$.} 
    If $S_\lambda$ is finite for every $\lambda\in\Lambda$, then $\Csp(\mA)$ is NP-hard by Lemma~\ref{lem:finite-nontrivial-hard}.
    \blue{Therefore, we can assume that $S^\circ_\lambda\neq\emptyset$ for some $\lambda\in\Lambda$, and the second compatibility condition implies that $S^{\circ}_\lambda$ is infinite for all $\lambda \in \Lambda$.}
    Let $d\geq 1$ be such that $S^\circ_\lambda$
    is a union of cosets of $d\Z$ for all $\lambda\in\Lambda$. Write $S^\circ_1 = \bigcup_{i=1}^n c_i+d\Z$, with $c_i\in\{0,\dots,d-1\}$.

    If $n\in\{2,\dots,d-1\}$, we claim that we can \blue{pp-define} a fully modular relation that is not Horn-definable.
    Indeed, let $\theta(x,y)$ be a formula that defines $\Family$.
    Note that $$\chi(x,y):=\theta(x,y)\land \theta(x,y+dx)\land\cdots\land\theta(x,y+\max(A\cup B)dx)$$ 
    holds precisely on the pairs $(\lambda,a)$
    such that $a\in S^\circ_\lambda$: since $x$ is forced to be in $\Lambda$ by $\theta$, a satisfying assignment gives a nonzero value $\lambda$ to $x$.
    Thus, if all of $y,y+d\lambda,\dots,y+\max(A\cup B)d\lambda$ are in $S_\lambda$, then they all must be in the modular part $S^\circ_\lambda$.
    The relation $T$ that $\chi$ defines is fully modular and is such that $T_\lambda = S^\circ_\lambda$ and in particular $T$ is not Horn-definable.
    It follows from Proposition~\ref{prop:hardness-fully-modular} that $\Csp(\mA)$ is NP-hard.
    
    Otherwise, the set $S^\circ_\lambda$ consists of a single coset of $d\Z$ for all $\lambda\in\Lambda$, and this coset is $\lambda c_1+d\Z$
    by the compatibility condition on $\Family$. Since $S^\circ_\lambda$ is assumed to not be Horn-definable, \blue{$A$ must contain an element $a$.}
    We claim that we can define another family of unary relations where the unary relations are finite and not singletons.
    Indeed, consider the formula
    \[\psi(x,y):= \exists z\left(\theta(x,y)\land\theta(x,z)\land y+z=(c_1+a)x\right)\]
    and let $T\subseteq\Z^2$ be the relation that it defines.
    First, note that $\psi(\lambda,c_1)$ and $\psi(\lambda,a)$ hold for all $\lambda\in\Lambda$,
    so that $|T_\lambda|>1$.
    We claim that $T_\lambda$ is finite.
    Since $A\cap S^\circ_1=\emptyset$, one has $a\neq c_1\bmod d$. Consequently, $c_1+a\neq 2c_1\bmod d$ and $(c_1+a)\lambda\neq 2c_1\lambda\bmod d$.
    The equation $y+z=(c_1+a)\lambda$ therefore forces that one of $y$ and $z$ is in $\lambda\cdot A$.
    Since $A$ is finite, there are only finitely many pairs satisfying this condition, thus showing that $1<|T_\lambda|<\infty$.
    It follows from Lemma~\ref{lem:finite-nontrivial-hard} that $\Csp(\mA)$ is NP-hard.
\end{proof}

\blue{As a corollary we obtain a simple-to-state condition implying that $\Csp(\mA)$ is NP-hard
(Corollary~\ref{cor:structure-endos}).
The corollary relies on the fact that $\End(\mA)$, being identified \blue{with} a subset of $\Z$, can be pp-defined in $\mA$. We prove this in the next lemma.}

\begin{lemma}\label{lem:pp-def-endos}
    \blue{Let $\mA$ be a finite-signature first-order reduct of $(\Z;+,1)$ that contains $+$}.
    Then the set $\End(\mA)$ has a \blue{pp-definition in $\mA$ that is additionally quantifier-free.}
\end{lemma}
\begin{proof}
    Let $E$ be the set of all the formulas $R(a_1\cdot x,\dots,a_r\cdot x)$
    for $R$ in the language of $\mA$ and $(a_1,\dots,a_r)\in R$.
    We then have that $\mA\models E(\lambda)$ iff $\lambda\in \End(\mA)$.
    We now show that there exists a finite subset $F\subseteq E$ that defines the same set of integers.

    For each relation $R$ of $\mA$, fix a standard definition $\phi_R$ in conjunctive normal form of $R$ in $(\Z;+,1)$.
    Let $M$ be the largest absolute value of a constant appearing in $\phi_R$.
    Consider the finite family $\mathcal F$ of equations $\sum \mu_ix_i = m$, where $\sum\mu_ix_i = m'$ is some equation appearing in $\phi_R$
    and $|m|\leq M$, together with all the equations $\sum\mu_ix_i = c\bmod d$ where $\sum\mu_ix_i = c'\bmod d$ is a modular equation appearing in $\phi_R$
    and $c\in\{0,\dots,d-1\}$.
    For each subset of $\mathcal F$ that is satisfiable by a tuple in $R$, pick a tuple $\tuple b\in R$ satisfying the formulas in this subset and add this tuple to a set $\mathcal S$.
    Repeat this operation for every relation of $\mA$, and let $\mathcal S$ be the finite set of tuples (of possibly different arities) that we obtain.
    Finally, let $F$ be the subset of $E$ where only the formulas associated with tuples from $\mathcal S$ are kept.

    We claim that $F$ defines $\End(\mA)$.
    Since $F\subseteq E$, it suffices to show that every $\lambda$ satisfying $F$ is an endomorphism of $\mA$.
    Let $\lambda\in\Z$ satisfy $F$, and let $\tuple a \in R$ be a tuple in some relation of $\mA$.
    Let $\tuple b\in \mathcal S$ be such that $\tuple b$ satisfies exactly the same equations in $\mathcal F$ as $\tuple a$.
    By construction, $\lambda\tuple b\in R$ so that in each clause of $\phi_R$, some equation is satisfied by $\lambda \tuple b$.
    We show that $\lambda\tuple a$ satisfies the same equations, so that $\lambda\tuple a\in R$.
    If $\lambda=0$, then $\lambda\tuple b=\lambda\tuple a$ so that $\lambda\tuple a\in R$. Suppose now that $\lambda\neq 0$.
    Let $\sum\mu_ix_i = c$ be a linear equation that is satisfied by $\lambda\tuple b$.
    Then necessarily $\lambda$ divides $c$, so that $\tuple b$ satisfies $\sum\mu_ix_i = \frac{c}{\lambda}$ and $|\frac{c}{\lambda}|\leq |c| \leq M$,
    so that $\sum\mu_i x_i = \frac{c}{\lambda}$ is an equation in $\mathcal F$.
    Consequently, $\tuple a$ also satisfies this equation and $\lambda\tuple a$ satisfies $\sum\mu_ix_i = c$.
    The proof for modular linear equations is similar. This proves that $\lambda$ is an endomorphism of $\mA$ and concludes the proof.
\end{proof}

\begin{lemma}\label{lem:pp-definition-one}
    \blue{Let $R\subseteq\Z$ be first-order definable over $(\Z;+,1)$ such that $(\Z;+,R)$ is a core.
    \begin{itemize}
        \item If $R^+ \neq\emptyset$, then $\{1\}$ or $\{1,-1\}$ is pp-definable in $(\Z;+,R)$.
        \item If $R^+=\emptyset$, then $R^-=\emptyset$ or $R=\Z\setminus\{0\}$.
    \end{itemize}}
\end{lemma}
\begin{proof}
Let $n$ be such that $R^{\circ}$ is a union of $n$ cosets of $d\Z$, i.e.,
\[ R^{\circ}=\bigcup_{i = 1}^n c_i + d\Z.\]
  Let us prove the first item.
  By Lemma~\ref{lem:pp-def-endos}, it suffices to prove that the only possible endomorphisms of the structure $(\Z;+,R)$ are $x\mapsto \lambda \cdot x$ with $\lambda \in\{1,-1\}$.
    Suppose that $x\mapsto \lambda \cdot x$ is an endomorphism. Then $\lambda\neq 0$ since the structure is a core, so suppose that $|\lambda|>1$.
    Let $a$ be the maximal element of $R^+$, and note that in particular $a+d\not\in R$ (it cannot be in $R^+$ because of the maximality assumption, and cannot be equal to any $c_i$ modulo $d$).
    Then $a\in R$, so $\lambda^q a\in R$ for all $q \in \N$. In particular, if $q$ is such that $\lambda^q>\max(R^+ \cup R^-)$ we obtain $\lambda^q a\in R^{\circ}$.
    This means that $\lambda^q a=c_i \bmod d$ for some $i\in\{1,\dots,n\}$.
    Finally, $\lambda^q (a+d) = \lambda^q a + \lambda^q d=c_i\bmod d$, so that $\lambda^q(a+d)\in R$. 
    This implies that $x\mapsto \lambda^q \cdot x$ is not an embedding, contradicting the core assumption on $(\Z;+,R)$.

    Let us now prove the second item.
    Let $b$ be some element of $R^-$.
    We must have $b=c_i \bmod d$ for some $i \in \{1,\dots,n\}$ \blue{since $R^- \subset R^{\circ}$}. 
    Note that the map $x\mapsto (d+1)x$ is an endomorphism of $(\Z;+,R)$, so it has to be an embedding.
    It follows that $(d+1)^m\cdot b\not\in R$ for any $m$. Suppose that $b$ is not $0$. Choose $m$ so that $(d+1)^m\cdot |b|>\max_{e\in R^-} |e|$
    so that $(d+1)^m\cdot b\not\in R^-$. But $(d+1)^m b = c_i\bmod d$, a contradiction.
    It follows that $R^-\subseteq\{0\}$, which concludes the proof.
\end{proof}
\begin{corollary}\label{cor:structure-endos}
    \blue{Let $\mA$ be a finite-signature first-order reduct of $(\Z;+,1)$ which contains $+$ and is a core.}
    If $\End(\mA)$ is not Horn-definable, then $\Csp(\mA)$ is NP-hard.
    Moreover, if $\End(\mA)$ is Horn-definable, then it is either $\{1\}$, $\Z\setminus\{0\}$, or $1+d\Z$ for some $d\geq 2$.
\end{corollary}

\ignore{\begin{proof}
    We first prove that $(\Z;+,\End(\mA))$ is a core.
    Let $x\mapsto \mu\cdot x$ be an endomorphism of $(\Z;+,\End(\mA))$.
    Then $x\mapsto \mu\cdot x$ must be an endomorphism of $\mA$: indeed, $1\in \End(\mA)$ so that $\mu\in \End(\mA)$.
    Since $\mA$ is a core, $x\mapsto k'\cdot x$ is a self-embedding of $\mA$, and therefore a self-embedding of $(\Z;+,\End(\mA))$
    since $\End(\mA)$ is pp-definable in $\mA$.
    
    Suppose now that $\End(\mA)$ is not Horn-definable.
    By Lemma~\ref{lem:pp-definition-one}, we obtain that $\End(\mA)$ can only be $\{1,-1\}$ or the union of cosets of $d\Z$, for some $d\geq 2$. 
    In these two cases, the family $\{\End(\mA)\}_{\lambda\in\End(\mA)}$ is clearly uniformly definable and satisfies the compatibility condition,
so that $\Csp(\mA)$ is NP-hard by Proposition~\ref{prop:unary-hard}.

    Finally, suppose that $\End(\mA)$ is Horn-definable.
    Then it is either a singleton, in which case it is $\{1\}$, or a coset $1+d\Z$ (for $d\geq 1$) from which a finite set $B$ of integers is removed.
    By Lemma~\ref{lem:pp-definition-one}, either $B=\emptyset$ (and then $d\geq 2$ because $0\not\in \End(\mA)$) or $\End(\mA)=\Z\setminus\{0\}$, which concludes the proof.
\end{proof}}
\begin{proof}
\blue{Lemma~\ref{lem:pp-def-endos} implies
that $R := \End(\mA)$ has a quantifier-free pp-definition in $\mA$. 
    We first prove that $(\Z;+,R)$ is a core.
    Indeed, let $\lambda$ be an endomorphism of $(\Z;+,R)$.
    Since $1 \in R$, we obtain that $\lambda \in R$, so that $x \mapsto \lambda x$ is a self-embedding of $\mA$
    by the fact that $\mA$ is a core.
    Since $R$ is has a quantifier-free definition over $\mA$, it follows that $x \mapsto \lambda x$ is also a self-embedding of $(\Z;+,R)$.}

\blue{First consider the case that $R$ is not Horn-definable. If $R^+ \neq \emptyset$ then 
Lemma~\ref{lem:pp-definition-one} implies that 
$\{1\}$ or $\{1,-1\}$ are pp-definable in $(\Z;+,R)$. 
All endomorphisms of $\mA$ must preserve this set, so $\End(\mA) = R = \{1,-1\}$ or
$\End(\mA) = R = \{1\}$; since $R$ is no Horn-definable, we must even have $R = \{1,-1\}$. 
But then the family $\{S_\lambda\}_{\lambda \in \{-1,1\}}$ with $S_{-1} = S_1 = \{1,-1\}$ is uniformly definable and compatible, the conditions being satisfied for $A = \{-1,1\}$ and
$B = \emptyset$: 
\begin{itemize}
\item $S_{1} =  A = S_{-1}$;
\item $S_1^{\circ} = S_{-1}^{\circ} = \emptyset$. 
\end{itemize}
Then Proposition~\ref{prop:unary-hard}
applied to $S_\lambda$ implies that $\Csp(\mA)$ is NP-hard. }

\blue{If $R^+ = \emptyset$ then Lemma~\ref{lem:pp-definition-one} implies that $R^- = \emptyset$ or 
that $R = \Z \setminus \{0\}$. In the latter case, $R$ would be Horn, contrary to the assumptions, so $R^{-} = \emptyset$. In this case, $R$ is fully modular, but not Horn definable, so NP-hardness
of $\Csp(\mA)$ follows from Proposition~\ref{prop:hardness-fully-modular}. 
This shows the first part of the statement.
}
 
 \blue{Finally, consider the case that $R$ is Horn-definable. If $R^+ = \emptyset$ then 
 Lemma~\ref{lem:pp-definition-one} implies that 
$R = \Z \setminus \{0\}$, and we are done, or $R^{-} = \emptyset$, in which case $R = 1 + d{\Z}$ for some $d \geq 2$ and we are also done.
Otherwise, $R^+ \neq \emptyset$ and 
 Lemma~\ref{lem:pp-definition-one} implies that 
 $\{1\}$ or $\{1,-1\}$ is pp-definable in $(\Z;+,R)$. 
}
\end{proof}

\subsection{Arbitrary arities}
We finally present the hardness proof in the general case where the structure contains a relation that is not Horn-definable. 
The strategy is to cut from a non-Horn relation $R$ a uniformly definable family $\Family$ of lines for which each $S_\lambda$ is not Horn-definable.
In a second step, we ensure that we get a family satisfying the compatibility condition, and we conclude using Proposition~\ref{prop:unary-hard}.
Call a formula $\phi$ in conjunctive normal form \emph{reduced} if removing any literal or clause from $\phi$
yields a formula that is not equivalent to $\phi$.
Note that minimal formulas are necessarily reduced.

\begin{lemma}\label{lem:step-1}
    \blue{Let $\mA$ be a finite-signature first-order reduct of $(\Z;+,1)$ which contains $+$ and is a core.
    Suppose that $\mA$ contains a relation $R$ that is not Horn-definable.}
    Then $\Csp(\mA)$ is NP-hard, or $\mA$ pp-defines a relation that is not Horn-definable and
    that has a minimal definition containing a non-Horn clause $\psi$ such that:
    \begin{itemize}
        \item no negated linear equation is in $\psi$,
        \item at least one linear equation is in $\psi$.
    \end{itemize}
\end{lemma}
\begin{proof}
Let $\phi$ be a standard minimal definition of $R$ in conjunctive normal form,
and let $\psi$ be a clause of $\phi$ that is not Horn.
From Corollary~\ref{cor:structure-endos}, we can suppose that $\End(\mA)$ is $\{1\}$, $\Z\setminus\{0\}$, or $1+d\Z$ for $d\geq 2$.
This implies that either $\{1\}$ is pp-definable or, by Proposition~\ref{prop:syntactic-core}, all the linear equations appearing in $\phi$ are homogeneous.

We can assume that $\psi$ does not contain any negative literal, per the assumption that $\phi$ is \blue{minimal}: 
indeed, 
\blue{consider the relation $R'$ defined by the formula
\begin{equation}\label{eq:not-Horn}
\phi' := \phi\land \sum\lambda_i x_i=c\tag{$\dagger$}
\end{equation} where $\sum\lambda_ix_i\neq c$ is in $\psi$. Either $c=0$,
in which case the relation $R'$ defined by (\ref{eq:not-Horn}) is pp-definable in $\mA$, or $\{1\}$ is pp-definable in $\mA$ and $R'$ is pp-definable in $\mA$, too. 
The relation $R'$ is not Horn-definable, and when we reduce the definition $\phi'$ of $R'$ we obtain a formula that has fewer literals in clauses that contain more than one literal, in contradiction to the minimality of $\phi$.} 

\blue{If $\psi$ contains a linear equation then we are done. Otherwise, $\psi$ only contains modular linear equations. If $\End(\mA) = \Z\setminus\{0\}$ then by Proposition~\ref{prop:syntactic-core} any modulus of a modular linear equation appearing in $\psi$ would have to be coprime with every nonzero integer, which is impossible. Therefore,} $\End(\mA)$ is $\{1\}$ or $1+d\Z$ for $d\geq 2$.
In the latter case, we can assume by Proposition~\ref{prop:syntactic-core} that all the modular linear equations in $\psi$ are modulo a divisor of $d$.
In the former case, let $d$ be a common multiple of all the moduli appearing in a modular linear equation in $\psi$.
Consider the structure $\mA/d\mA$ defined in Proposition~\ref{prop:hardness-fully-modular}.
\blue{The relation $T$ obtained from $R$ in this structure is not a coset of a subgroup $H$ of $(\Z/d\Z)^k$  (where $k$ is the arity of $R$): otherwise this coset is definable by a conjunction $\theta$ of modular linear equations modulo a divisor of $d$. Replacing $\psi$ by $\theta$ in $\phi$ would produce a smaller definition of $R$, a contradiction to the minimality of $\phi$.}
Moreover, $(\mA/d\mA,1)$ is pp-interpretable in $\mA$: in the two cases that $\End(\mA)=\{1\}$  and $\End(\mA)=1+d\Z$, the preimage of $\{1\}$
under the canonical projection \blue{$x \mapsto x\bmod d$} is pp-definable in $\mA$.
We conclude as in Proposition~\ref{prop:hardness-fully-modular} that $\Csp(\mA)$ is NP-hard.
\end{proof}

\begin{theorem}\label{thm:not-Horn-hard}
    \blue{Let $\mA$ be a finite-signature first-order reduct of $(\Z;+,1)$ 
    which contains $+$ and is a core.}
    Suppose that $\mA$ contains a relation that is not Horn-definable.
    Then $\Csp(\mA)$ is NP-hard.
\end{theorem}
\begin{proof}
From Lemma~\ref{lem:step-1}, we can suppose that $\mA$ pp-defines a relation $R$ that is not Horn-definable,
that has a reduced standard definition $\phi$ containing a non-Horn clause $\psi$ with at least one linear equation $(L)$ and no negated linear equation.
Since $\psi$ is not Horn, it contains at least another equation $(L')$, possibly modular.
Let $(a_1,\dots,a_n)$ satisfy $\phi$ and only $(L)$ in $\psi$. Such a tuple exists by the assumption that $\phi$ is reduced.
Similarly, let $(b_1,\dots,b_n)$ satisfy $\phi$ and only $(L')$ in $\psi$.
Let $S$ be the binary relation such that $(\lambda,t)\in S$ if, and only if,
$\lambda\in\End(\mA)$ and $t(\tuple a-\tuple b)+\lambda\tuple b$ is in $R$.
Note that $\tuple a$ and $\tuple b$ being fixed, $S$ is pp-definable over $\mA$.
Therefore, we obtain a family $\Family$ that is uniformly definable in $\mA$, where $\Lambda=\End(\mA)$. \blue{Clearly, $1 \in \Lambda$, and $\Lambda \subseteq \Z \setminus \{0\}$ and $0,\lambda\in S_\lambda$}. Moreover, note that 
\begin{equation}\label{eq:t1-tk}
    S_\lambda\cap \lambda\cdot\Z = \lambda\cdot S_1\tag{$\ddag$}
\end{equation}
holds for all $\lambda\in \End(\mA)$.
Indeed:
\begin{align*}
    t\in S_1 &\Leftrightarrow t(\tuple a-\tuple b)+\tuple b\in R\\
             &\Leftrightarrow \lambda t(\tuple a-\tuple b)+\lambda\tuple b\in R & \text{because $\mA$ is a core}\\
             &\Leftrightarrow \lambda t\in S_\lambda.
\end{align*}

We prove that for all $\lambda\in \End(\mA)$ the relation $S_\lambda$ is not Horn-definable.
Since $0,\lambda\in S_\lambda$, it suffices to prove that $S_\lambda$ omits infinitely many multiples of $\lambda$,
and by (\ref{eq:t1-tk}) it suffices to prove that $\ell\not\in S_1$ for infinitely many $\ell$.
Let $\ell$ be such that $\ell=1\bmod d'$, for every modulus $d'$ appearing in $\psi$.
\blue{We claim} that $\ell(\tuple a-\tuple b) + \tuple b$ does not satisfy any modular linear equation in $\psi$. Indeed, let $\sum_i\sigma_ix_i=c\bmod d'$ be such a
modular linear equation. Then we have
\begin{align*}
    \sum_i\sigma_i(\ell(a_i-b_i) + b_i) = c\bmod d' &\Leftrightarrow \sum_i\sigma_i a_i = c\bmod d',
\end{align*}
which is a contradiction to the choice of $\tuple a$ since $\tuple a$
only satisfies $(L)$ in $\psi$ and $(L)$ is assumed to be non-modular.
Consider now a linear equation $\sum\sigma_ix_i=c$ in $\psi$. \blue{This equation is satisfied by $\ell(\tuple a-\tuple b)+\tuple b$ if, and only if
\begin{equation}\label{eq:not-Horn-hard}
    \ell\cdot\sum_i \sigma_i (a_i-b_i) = c - \sum_i \sigma_i b_i \, . \tag{$\star$}
\end{equation}
Suppose first that $\sum\sigma_i(a_i-b_i) = 0$. Then (\ref{eq:not-Horn-hard}) is satisfied if, and only if, we have $\sum\sigma_i b_i = c=\sum\sigma_i a_i$.
This implies that both $\tuple a$ and $\tuple b$ satisfy the equation; this is a contradiction to our choice of the vectors $\tuple a$ and $\tuple b$,
so that $\ell(\tuple a-\tuple b) + \tuple b$ does not satisfy (\ref{eq:not-Horn-hard}).
Suppose now that $\sum\sigma_i(a_i-b_i)\neq 0$.
If $\ell> |c - \sum\sigma_i b_i|$, it is then clear that (\ref{eq:not-Horn-hard}) is not satisfied.
Therefore, for infinitely many $\ell$, the tuple $\ell(\tuple a-\tuple b) + \tuple b$
does not satisfy any literal in $\psi$ and $\ell\not\in S_1$.}

\blue{Let $\theta(x,y)$ be a minimal reduced standard definition of $S$}.
By inspection of the formula $\theta$, one finds that
$S^+_\lambda$ and $S^-_\lambda$
consist of points of the form $-\frac{a\lambda}{b}$, where $ax+by=0$ is an equation in $\theta$.
Note that since $a$ and $b$ are taken coprime by the minimality of $\theta$, if $b$ divides $a\lambda$ then $b$ divides $\lambda$.
Let $m:=\text{lcm}\{|b| : ax+by=0\text{ is an equation in $\theta$}\}$, and note that \blue{$m\Z\cap \Lambda$} is not empty by the previous remark.
Let $T\neq\emptyset$ be the binary relation defined by $\theta(m\cdot x, y)$. 
\blue{For all $\lambda \in \End(\mA)$, the set} $T_\lambda$ is not Horn-definable and of the form $(T^\circ_\lambda\cup (\lambda\cdot P))\setminus (\lambda\cdot Q)$,
where $P$ and $Q$ are finite sets that are independent of $\lambda$ and $T^\circ_\lambda=S^\circ_{m\lambda}$.
For the family of relations $\{T_\lambda\}_{\lambda\in\End(\mA)}$ to satisfy the compatibility condition,
it remains to prove that $c+d\Z\subseteq T^\circ_1$ if, and only if, $\lambda c+d\Z\subseteq T^\circ_\lambda$, for all $d\geq 1$ and $c\in\{0,\dots,d-1\}$.
Suppose that $c+d\Z\subseteq T^\circ_1$ for some $d\geq 1$ and \blue{suppose that} the cosets of $T^\circ_\lambda$ are cosets of $d'\Z$.
By Proposition~\ref{prop:syntactic-core}, $\lambda$ and $d'$ are coprime. Therefore, there exists $\mu\in\Z$ such that $\lambda\mu=1\bmod d'$.
We have $c+d\mu\Z\subseteq T^\circ_1$,  because $d$ divides $d\mu$.
It follows that $\lambda c+\lambda d\mu\Z\subseteq T^\circ_\lambda$.
Now, let $x\in \lambda c+d\Z$, say $x=\lambda c+qd$.
Then we have $x=\lambda c+q\lambda\mu d\bmod d'$.
Note that $\lambda c+q\lambda\mu d\in \lambda c+\lambda d\mu\Z\subseteq T^\circ_\lambda$.
Since the cosets in $T^\circ_\lambda$ are cosets of $d'\Z$, we obtain that $x\in T^\circ_\lambda$ and consequently that $\lambda c+d\Z\subseteq T^\circ_\lambda$.
Conversely, if $\lambda c+d\Z\subseteq T^\circ_\lambda$ then $\lambda c+d\lambda\Z\subseteq T^\circ_\lambda$, because $d$ divides $d\lambda$. 
Since $\mA$ is a core, 
$x\mapsto \lambda\cdot x$ is a self-embedding of $\mA$, so that $c+d\Z\subseteq T^\circ_1$.

To conclude, the family of compatible relations $\{T_\lambda\}_{\lambda\in\End(\mA)}$ is uniformly pp-definable in $\mA$ and consists of non-Horn relations.
By Proposition~\ref{prop:unary-hard}, we obtain that $\Csp(\mA)$ is NP-hard.
\end{proof}

We illustrate our proofs in some examples below.
\begin{example}
Consider the binary relation \[S=\{(\lambda,t) \mid (\lambda=1\bmod 4\land t=1\bmod 4) \lor (\lambda=3\bmod 4\land t=3\bmod 4) \lor t=0\}.\]
One sees that the set of endomorphisms of $\mA:=(\Z;+,S)$ is equal to $\End(\mA):=1+2\Z$. 
Moreover, $S$ is not Horn-definable.
For every $\lambda\in\End(\mA)$, one has $S_\lambda = \{0\}\cup (\lambda+4\Z)$.
When $\lambda$ is fixed, one can define a finite set by $\exists y(x\in S_\lambda\land y\in S_\lambda\land x+y=\lambda)$, which defines $\{0,\lambda\}$.
One then obtains a reduction from \OneInThree\ by $\exists x,y,z\in\{0,\lambda\}: x+y+z\in\{0,\lambda\}$.
Finally, by existentially quantifying over \blue{$\lambda \in \End(\mA)$ we} obtain a reduction from $\OneInThree$ to $\Csp(\Z;+,S)$.
\end{example}

\begin{example}
	Let $R:=\{0\}\cup (1+3\Z)\cup (2+3\Z)$ and $K=1+3\Z$.
	Note that Proposition~\ref{prop:unary-hard} does not apply to $\Csp(\Z;+,R)$ since $(\Z;+,R)$
	is not a core \blue{(we have $0 \in \End(\Z;+,R)$)}. 
	\blue{Neither} does Corollary~\ref{cor:structure-endos} \blue{apply} to $\Csp(\Z;+,R,K)$ since
	$\End(\Z;+,R,K) = K$, \blue{which is clearly Horn-definable in $(\Z;+,R,K)$}.
	\blue{But one obtains hardness of $\Csp(\Z;+,R,K)$ by Theorem~\ref{thm:not-Horn-hard}. Indeed,}  
	pick $a=0$ (satisfying the linear equation $x=0$ in the definition of $R$) and $b=1$ (satisfying the modular linear equation $x=1\bmod 3$ in the definition of $R$),
	and define the relation $S=\{(\lambda,t)\mid \lambda\in K\land \lambda-t\in S\}$.
	Note that for all $\lambda\in K$, we have $S_\lambda = \{\lambda\}\cup 3\Z\cup (2+3\Z)$.
	The formula $\exists w(w\in K\land S(\lambda,t)\land S(\lambda,t+3w))$ defines the relation
	$T=\{(\lambda,t)\mid \lambda=1\bmod 3 \land (t=0\bmod 3\lor t=2\bmod 3)\}$, which is fully modular and not Horn-definable.
	Proposition~\ref{prop:hardness-fully-modular} implies that $\Csp(\Z;+,R,S)$ is NP-hard.
\end{example}

\section{Tractability} \label{sec:tract}

In this section we show the following. 

\begin{proposition}\label{prop:tract}
Let $\mA$ be a structure with finite relational signature, domain $\mathbb Z$, and whose
relations have quantifier-free Horn definitions over  
$(\mathbb Z;+,1)$.
Then there is an algorithm that solves CSP$(\mA)$ in polynomial time.
\end{proposition}

This result follows from the following more general 
result. 

\begin{theorem}\label{thm:tract}
Let $\phi$ be a quantifier-free Horn formula over $({\mathbb Z};+)$, allowing parameters from ${\mathbb Z}$ represented in binary. Then there exists a polynomial-time algorithm to decide whether $\phi$ is satisfiable over $(\mathbb Z;+)$. 
\end{theorem}

The proof of Theorem~\ref{thm:tract} can be found
at the end of this section. We first show how to derive Proposition~\ref{prop:tract}. 

\begin{proof}[Proof of Proposition~\ref{prop:tract}]
The input of CSP$(\mA)$ consists of a primitive positive sentence whose atomic formulas
are of the form $R(x_1,\dots,x_k)$ where $R$ is quantifier-free Horn definable over $\mathcal{L}_{(\mathbb Z;+,1)}$. Since $\sum_{i=1}^n a_ix_i=b\bmod c$ is equivalent to  $\sum_{i=1}^n a_ix_i=b+ ck$, where $k$ is a new integer variable, we can as well assume that the input
to our problem consists of a set of Horn clauses over $(\mathbb Z;+,1)$. This is tacitly the process of quantifier introduction, the converse of quantifier elimination. Then apply Theorem~\ref{thm:tract}. 
\end{proof} 

Our algorithm for the proof of Theorem~\ref{thm:tract} uses two other well-known algorithms:
\begin{enumerate}
\item a polynomial-time algorithm for satisfiability of linear diophantine equations, i.e.,
the subproblem of the computational problem from
Theorem~\ref{thm:tract} where the input only contains atomic formulas (see, e.g.,~\cite{Schrijver}).  
\item a polynomial-time algorithm to compute
the rank of a matrix over ${\mathbb Q}$; this allows us in particular to decide whether a given linear system of equalities implies another equality over the rationals (this is standard, using Gaussian elimination; again, see~\cite{Schrijver} for a discussion of the complexity). 
\end{enumerate}

These two algorithms can be combined to obtain the following.

\begin{lemma}\label{lem:implies}
There is a polynomial-time algorithm that decides
whether a given system $\Phi$ of linear diophantine equations implies another given diophantine equation $\psi$ 
over ${\mathbb Z}$.  
\end{lemma}
\begin{proof}
First, use the first algorithm above to test whether $\Phi$ has a solution over ${\mathbb Z}$. If no, return yes (false implies everything).
If yes, we claim that $\Phi$ implies $\psi$ 
over ${\mathbb Q}$ (which can be tested by the second algorithm above) if and only if $\Phi$ implies $\psi$ over the integers. Clearly, 
if every rational solution of $\Phi$ satisfies $\psi$,
then so does every integer solution. Suppose now that there exists a rational solution $\alpha$ to $\Phi$ which does not satisfy $\psi$. Also take an integer solution $\beta$ to $\Phi$. \blue{Then on the line
$L$ that goes through $\alpha$ and $\beta$ there are infinitely many integer points. 
If infinitely many points on a line satisfy $\psi$, then
all points of the line must satisfy $\psi$. Since $\alpha \in L$ does not satisfy $\psi$ it follows that 
an integer point on $L$ does not satisfy $\psi$, i.e., 
 $\Phi$ does not imply $\psi$ over the integers.}
\end{proof}

\ignore{
The first polynomial-time algorithms for the satisfiability of linear diophantine equation systems have been discovered by Frumkin and, independently Sieveking and von zur Gathen. 
Kannan and Bachem~\cite{KannanBachem} presented a method based on first computing the Hermite Normal Form 
of the matrix given by the linear system.
See discussion in the text-book Schrijver~\cite{Schrijver}. 
Further improvements have been made by  Chou and Collins~\cite{ChouCollins}, and Storjohann (1998). The best known algorithm has 
space complexity $O(n^2 \log M)$ 
(essentially the same size as the input)
and worst-case running time in $O(n^5 polylog(n,M))$ (Micciancio and Warinschi'2001). 
}

Given the two mentioned algorithms, 
our procedure for the proof of Theorem~\ref{thm:tract} is basically an implementation of positive unit clause resolution. It takes the same form as the algorithm presented in~\cite{HornOrFull} for satisfiability over the rationals. 

\begin{figure}[h]
\begin{center}
\small
\fbox{
\begin{tabular}{l}
{\rm // Input: a set of Horn-clauses $\mathcal C$
 over $({\mathbb Z};+)$ with parameters.} \\
{\rm // Output: \emph{satisfiable} if $\mathcal C$ is satisfiable in $({\mathbb Z};+)$, \emph{unsatisfiable} otherwise} \\
Let $\cal U$ be clauses from $\cal C$ that only contain a single positive literal. \\
If $\cal U$ is unsatisfiable then return \emph{unsatisfiable}. \\
Do \\
\hspace{.5cm}    For all negative literals $\neg \phi$ in clauses from $\cal C$ \\
\hspace{1cm}    If $\cal U$ implies $\phi$, then delete the negative literal $\neg \phi$ from all clauses in $\cal C$. \\
\hspace{.5cm}    If $\cal C$ contains an empty clause, then return \emph{unsatisfiable}. \\
\hspace{.5cm}    If $\cal C$ contains a clause with a single positive literal $\psi$, then add $\{\psi\}$ to $\cal U$. \\
Loop until no literal has been deleted \\
Return \emph{satisfiable}.
\end{tabular}}
\end{center}
\caption{An algorithm for satisfiability of Horn formulas with parameters over $({\mathbb Z};+)$.}
\label{fig:alg}
\end{figure}

\begin{proof}
We follow the proof of Proposition 3.1 from \cite{HornOrFull}. We first discuss the correctness of the algorithm.

When $\cal U$ logically implies $\phi$ (which can be tested with the algorithm from Lemma~\ref{lem:implies}) 
then the negative literal $\neg \phi$ is never satisfied and can be deleted from all clauses without affecting the set of solutions.
Since this is the only way in which literals can be deleted from clauses, it is clear that if one clause becomes empty the instance
is unsatisfiable.

If the algorithm terminates with \emph{satisfiable}, then no negation of an inequality
is implied by $\cal U$. If $r$ is the rank of the linear equation system defined by $\cal U$,
we can use Gaussian elimination to eliminate $r$ of the variables from all literals in the remaining clauses.
For each of the remaining inequalities, consider the
sum of absolute values of all coefficients.
Let $S$ be one plus the maximum of this sum
over all the remaining inequalities. Then setting the $i$-th
variable to $S^i$ satisfies all clauses.
To see this, take any inequality, and assume that $i$ is the highest variable
index in this inequality. Order the inequality in such a way that the variable
with highest index is on one side and all other variables on the other side of the
$\neq$ sign. The absolute value on the side with the $i$-th variable is at
least $S^i$. The absolute value on the other side is less than $S^i - S$, since all
variables have absolute value less than $S^{i-1}$ and the sum of all
coefficients is less than $S-1$ in absolute value. Hence, both sides of the
inequality have different absolute value, and the inequality is satisfied.
Since all remaining clauses have at least one inequality, all
constraints are satisfied.

Now let us address the complexity of the algorithm.
With appropriate data structures, the time needed for removing negated literals $\neg \phi$
from all clauses when $\phi$ is implied by $\cal U$ is linearly bounded in the input size since each literal can be removed at most once.
\end{proof}

\section{Conclusion}
We are finally in position to prove 
the main result. 

\ignore{\begin{proof}[Proof of Theorem~\ref{thm:main}]
Let $\mA$ be a finite-signature first-order reduct of $({\mathbb Z};+,1)$ that contains $+$ and is a core. 
Membership in NP of CSP$(\mA)$ has been noted in the introduction. 
If some of the relations of $\mA$ is not definable by Horn formulas, then 
NP-hardness follows from Theorem~\ref{thm:not-Horn-hard}. Otherwise, polynomial-time tractability follows from Proposition~\ref{prop:tract}.
\end{proof}}

\begin{proof}[Proof of Theorem~\ref{thm:dichotomy}]
\blue{Let $\mA$ be a finite-signature first-order reduct of $({\mathbb Z};+,1)$ that $\mA$ contains $+$.} 
By Lemma~\ref{lem:core} there exists a core 
$\mB$ of $\mA$. If $\mB$ has only one element then
$\Csp(\mB)$ and $\Csp(\mA)$ are trivially in P. 
Otherwise, $\mB$ is itself first-order definable in
$({\mathbb Z};+,1)$ and contains $+$ \blue{by Lemma~\ref{lem:core}}, 
and 
the statement follows from Theorem~\ref{thm:not-Horn-hard} and Proposition~\ref{prop:tract}. 
\end{proof}

\bibliography{local,local2}
\end{document}